\documentclass[12pt]{amsart}

\usepackage{graphicx,color}
\usepackage{setspace}
\usepackage{amsmath, amssymb,amsthm, amsfonts}

\usepackage{accents}

\newtheorem{thm}{Theorem}[section]
\newtheorem{lma}{Lemma}[section]

\newtheorem{cor}{Corollary}[section]

\theoremstyle{definition}
\newtheorem{definition}{Definition}[section]

\theoremstyle{remark}
\newtheorem{remark}{Remark}[section]

\numberwithin{equation}{section}

\newcommand{\tr}{\mbox{tr}}

\renewcommand{\div}{\mbox{div}}

\newcommand{\Ric}{\mbox{Ric}}
\newcommand{\R}{\mathbb R}

\newcommand{\be}{\begin{equation}}
\newcommand{\ee}{\end{equation}}
\newcommand{\bee}{\begin{equation*}}
\newcommand{\eee}{\end{equation*}}

\def\p{\partial}

\def\lf{\left}
\def\ri{\right}
\def\Pi{\displaystyle{\mathbb{II}}}

\def\fe{\mathfrak{e}}

\def\S{\Sigma}

\def\H{\mathbb{H}}

\def\e{\epsilon}

\def\a{\alpha}

\def\wt{\widetilde }

\def\ol{\overline}

\begin{document}
\title[Some Conformal Positive Mass Theorems]{
Some Conformal Positive Mass Theorems}

\author{Luen-Fai Tam$^1$}
\address[Luen-Fai Tam]{The Institute of Mathematical Sciences and Department of
 Mathematics, The Chinese University of Hong Kong, Shatin, Hong Kong, China.}
 \email{lftam@math.cuhk.edu.hk}
\thanks{$^1$Research partially supported by Hong Kong RGC General Research Fund \#CUHK 14305114}
\author{Qizhi Wang}
\address[Qizhi Wang]{The Institute of Mathematical Sciences and Department of
 Mathematics, The Chinese University of Hong Kong, Shatin, Hong Kong, China.}
 \email{qzwang@math.cuhk.edu.hk}
\renewcommand{\subjclassname}{
  \textup{2010} Mathematics Subject Classification}
\begin{abstract}
In \cite{S}, Simon proved a conformal positive mass theorem, which was used to prove uniqueness of black holes \cite{MS,G}. In this note, we will generalize Simon's conformal positive mass theorem   in two directions. First we will consider spacetime version of conformal positive mass theorems on asymptotically flat initial data set. Next, we will prove a  conformal positive mass theorem on  asymptotically hyperbolic manifolds.
\end{abstract}
\renewcommand{\subjclassname}{
  \textup{2010} Mathematics Subject Classification}

\date{July 2016}
\maketitle
\markboth{  Luen-Fai Tam and Qizhi Wang}
{Some Conformal Positive Mass Theorems}
\section{introduction}

 In \cite{S}, motivated by the work of Masood \cite{Masood}, Simon proved two conformal positive mass theorems. The following version has been used to prove some uniqueness results for charged black holes, see \cite{MS,G}: Suppose $(M^n,g)$ is asymptotically flat manifolds. For simplicity, let us assume that $M$ has only one end. If $n>3$, we also assume that $M$ is spin. Suppose $f$ is a smooth function so that $e^{2f}g$ is also asymptotically flat. If the scalar curvature $\mathcal{S}$ of $g$ and the scalar curvature $\wt{\mathcal{S}}$ of $\wt g=e^{2f}g$ satisfies
$$
\mathcal{S}+\a e^{2f}\wt{\mathcal{S}}\ge0
$$
for some $\a>0$, then the ADM masses $\mathfrak{m}(g)$ of $g$ and $\mathfrak{m}(\wt g)$ of $\wt g$ satisfies
$$
\mathfrak{m}(g)+\a \mathfrak{m}(\wt g)\ge 0
$$
Moreover, equality holds if and only if $g, \wt g$ are the same and are flat.
In \cite{Wa}, the second author generalized the above result to asymptotically flat manifolds with compact inner boundary.

In this note, we will generalize Simon's result in two directions. First we will consider spacetime version of conformal positive mass theorems on asymptotically flat initial data set. Next, we will prove a   conformal positive mass theorems  asymptotically hyperbolic manifolds. We always assume the dimension of the manifold is at least three. We obtain the following:
\begin{thm}\label{t-intro-1}
Let $(M^n, g, K)$ be an asymptotically flat initial data set. If $n>3$, we also assume the manifold is spin. Let $f$ be  a smooth function which decays in an appropriate way so that $(M^n, e^{2\beta f}g, e^{\beta f}K)$ is also an asymptotically flat initial data set for all $0<\beta\le 1$.  Let $\wt g=e^{2\beta f}g$, $\wt K= e^{\beta f}K$. Suppose
$$
(1-\beta)\mu+\beta e^{2f}\wt \mu\ge (1-\beta)|J|_g+\beta e^{2f}|\wt J|_{\wt g}.
$$
Then the energy-momentum vector $(E,P)$ for $g$ and the energy-momentum vector $(\wt E,\wt P)$ satisfies
$$
(1-\beta)E+\beta \wt E\ge |P|.
$$
Moreover, if $(1-\beta)E+\beta \wt E=0$, then $f=0$ and $(M,g)$ can be isometrically embedded in the Minkowski space as an asymptotically flat Cauchy surface with second fundamental form $K$.
\end{thm}
Here as usual
\bee
\left\{
  \begin{array}{ll}
   2\mu:=&\mathcal{S}-|K|_g^2+\lf(\tr_gK\ri)^2; \\
    J^i:=&\nabla_j(K^{ij}-(\tr_gK)g^{ij}).
  \end{array}
\right.
\eee
$\wt\mu, \wt J$ are defined similarly for $\wt g, \wt K$.  For more precise statement and definitions of asymptotically flat initial data set, see section \ref{spacetime cpmt}.

We also obtain a positive mass theorem for asymptotically hyperbolic spaces.
\begin{thm}\label{t-intro-2} Let $(M^n,g)$ be an asymptotically hyperbolic space. If $n>3$, we assume the manifold is spin. Let $f$ be a smooth function so that $e^{2\beta f}g$ is also asymptotically hyperbolic. If the scalar curvature $\mathcal{S}$ of $g$ and $\wt{\mathcal{S}}$ of $e^{2f}g$, satisfy
$$
e^{-2\beta f}\lf((1-\beta)\mathcal{S}+\beta e^{2f}\wt{\mathcal{S}}\ri)\ge -n(n-1)
$$
then the mass integral $\mathbf{M}(g)$ of $g$ and the mass integral $\mathbf{M}(\wt g)$ satisfy $(1-\beta)\mathbf{M}(g)+\beta\mathbf{M}(\wt g)$ is future time like or zero. If it is zero, then $f=0$ and $(M,g)$ is isometric to the hyperbolic space.
\end{thm}
Again the precise definitions of asymptotically hyperbolic space and the mass integral are in section \ref{AH cpmt}.
We also prove a related results for asymptotically  hyperbolic spaces with inner boundary. See section \ref{AH cpmt} for more details.

Our proofs are just  applications to known results on various positive mass theorems \cite{SchoenYau,Witten,PT,ChruscielMaerten2006,BeigChruciel,Wang,ChruscielHerzlich,H}. What we have done is to compute various quantities so that one can apply known results directly.
We will point out what kind of positive mass theorems we use in appropriate places.

The organization of this note is as follows: In section \ref{spacetime cpmt} we discuss spacetime version of conformal positive mass theorems for asymptotically flat initial data sets. In section \ref{AH cpmt}, we discuss conformal positive mass theorems on asymptotically hyperbolic spaces.

{\it Acknowledgement}: The first author would like to thank Naqing Xie for some useful discussions.

\section{spacetime conformal positive mass theorem for asymptotically flat manifolds}\label{spacetime cpmt}

In this section, we obtain a spacetime version of conformal positive mass theorem on asymptotically flat (AF) manifolds.
First we recall the following facts, see \cite{LeeParker1987} for example.
\begin{lma}\label{l-conformal-1} Let $(M^n,g)$ be a Riemannian manifold. Let $f$ be a smooth
function on $M$ and let $\ol g=e^{2f}g$. Then
\bee
\left\{
  \begin{array}{ll}
    \ol\Gamma_{ij}^k-\Gamma_{ij}^k=&
  f_j\delta_i^k+f_i\delta_j^k-f_lg^{kl}g_{ij} \\
   \ol R_{ij}-R_{ij}=& -g_{ij}\Delta_gf+(2-n)f_{;ij}+(2-n)|\nabla_g f|^2g_{ij}+(n-2)f_if_j . \\
    e^{2f}\ol {\mathcal{S}}=&\mathcal{S}-2(n-1)\Delta_gf -(n-1)(n-2)|\nabla_g f|^2.
  \end{array}
\right.
\eee
Here $\Gamma_{ij}^k, \ol \Gamma_{ij}^k$ are Christoffel symbols with respect to local coordinates; $R_{ij},\ol R_{ij}$ are Ricci tensors; $ \mathcal{S}, \ol {\mathcal{S}}$ are scalar curvatures of $g, \ol g$ respectively. $f_{;ij}$ is the Hessian of $f$ with respect to $g$.
\end{lma}

\begin{definition}\label{d-initialdata-1} $(M^n,g,K)$ is said to be an asymptotically flat initial data set if $g$ is a smooth complete metric on $M$, $K$ is smooth a symmetric (0,2)   tensor  such that outside a compact set $M$ consists of finitely many ends. Moreover each end $N$ is diffeomorphic to $\R^n\setminus B(R)$ for some $R$ so that in the standard coordinates of $\R^n$, near infinity
$$
|g_{ij}-\delta_{ij}|+r|\p_k g_{ij}|+r|K_{ij}|\le Cr^{-\tau}
$$
for some $\tau>\max(1/2, n-3)$ for all $i, j, k$, where $r=|x|$ is the Euclidean distance from the origin. In addition, near infinity
$$
|\mu|, |J|_g\le C r^{-n-\e}
$$
for some $\e>0$, where
\be\label{e-muJ-1}
\left\{
  \begin{array}{ll}
   2\mu:=&\mathcal{S}-|K|_g^2+\lf(\tr_gK\ri)^2; \\
    J^i:=&\nabla_j(K^{ij}-(\tr_gK)g^{ij}).
  \end{array}
\right.
\ee
\end{definition}

Let $(M,g,K)$ be as in the definition.
Let $f$ be a smooth function and let $0<\beta\le 1$. Consider the metric $\ol g=e^{2\beta f}g$ and $\ol K=e^{\beta f}K$. Define $\ol \mu, \ol J$ as in the \eqref{e-muJ-1} with respect to $\ol g, \ol K$. Then we have the following relation. 
\begin{lma}\label{l-muJ-1} With the above notations,
$$
2\ol\mu=e^{-2\beta f}\lf(2\mu+\lf\{-2(n-1)\beta\Delta_g f-(n-1)(n-2)\beta^2|\nabla_gf|^2\ri\}\ri),
$$
and
$$
e^{3\beta f}\ol J^i=J^i- \beta \lf((n-1)f_jK^{ij} -2f_lg^{il}\tr_gK\ri).
$$
\end{lma}
\begin{proof} Let $\ol{\mathcal{S}}$ be the scalar curvature of $\ol g$, then by Lemma \ref{l-conformal-1},
\bee
\begin{split}
2e^{2\beta f}\ol\mu=& e^{2\beta f}\lf(\ol{\mathcal{S}}-|\ol K|^2_{\ol g}+(\tr_{\ol g}\ol K)^2 \ri)\\
=&\mathcal{S}+\lf\{-2(n-1)\beta\Delta_g f-(n-1)(n-2)\beta^2|\nabla_gf|^2\ri\}-|K|^2_g+(\tr_{g}K)^2\\
=&2\mu+ \lf\{-2(n-1)\beta\Delta_g f-(n-1)(n-2)\beta^2|\nabla_gf|^2\ri\}.
\end{split}
\eee
Let $\ol\nabla$ be the covariant derivative with respect to $\ol g$.
\be\label{e-J-1}
\begin{split}
\ol\nabla_j\ol K^{ij}=&\ol\nabla_j(e^{-3\beta f} K^{ij})\\
=&  e^{-3\beta f}\lf(-3\beta f_j K^{ij}+\nabla_j K^{ij}+(\ol \nabla_j-\nabla_j)K^{ij}\ri).
\end{split}
\ee
In local coordinates, by Lemma \ref{l-conformal-1}, let
$$
A_{ij}^k:=\ol\Gamma_{ij}^k-\Gamma_{ij}^k= \beta \lf(f_j\delta_i^k+f_i\delta_j^k-f_lg^{kl}g_{ij}\ri).
$$
 Then
\bee
\begin{split}
(\ol \nabla_j-\nabla_j)K^{ij}=& A_{jk}^iK^{kj}+A_{jk}^jK^{ik}\\
=&\beta\lf(f_j\delta_k^i+f_k\delta_j^i-f_lg^{il}g_{kj}\ri)K^{kj}
+\beta\lf(f_k\delta_j^j+f_j\delta_k^j-f_lg^{jl}g_{jk}\ri)K^{ik}\\
=&\beta\lf(f_jp^{ij}+f_kp^{ik}-f_lg^{il}\tr_gK\ri)+\beta\lf(nf_kK^{ik}+f_jK^{ij}-f_lK^{il}\ri)\\
=&\beta\lf((n+2)f_jK^{ij} -f_lg^{il}\tr_gK \ri).
\end{split}
\eee
By \eqref{e-J-1}, we have
\be\label{e-J-2}
e^{3\beta f}\ol\nabla_j\ol K^{ij}=\nabla_jK^{ij}-\beta\lf((n-1)f_jK^{ij}-f_lg^{il}\tr_gK\ri).
\ee

On the other hand,
\be\label{e-J-3}
\begin{split}
\ol\nabla_j\lf(\tr_{\ol g} \ol K) \ol g^{ij}\ri)=&(\tr_{\ol g}\ol K)_j \ol g^{ij}\\
=& \lf(e^{-\beta f}\tr_gK\ri)_je^{-2\beta f}g^{ij} \\
=&e^{-3 \beta f}\lf(-\beta f_j\tr_gK+(\tr_gK)_j\ri)g^{ij}
\end{split}
\ee
By \eqref{e-J-2}, \eqref{e-J-3}, we have
\bee
\begin{split}
e^{3\beta f}\ol J^i=J^i- \beta \lf((n-1)f_jK^{ij} -2f_lg^{il}\tr_gK\ri).
\end{split}
\eee
\end{proof}

 \begin{cor}\label{c-AFinitial-1} Suppose $(M^n,g,K)$ is an AF initial data set and let $\tau$ and $\e>0$ be as in Definition \ref{d-initialdata-1}.  Let $f$ be a smooth function on $M$.  Suppose at each end
 $$
  |f|+r|\nabla_gf|+r|\Delta_gf|^\frac12\le Cr^{-\tau}
  $$
  for some $C>0$ where $r=|x|$ and $x$ is the coordinates of an end in the definition of AF initial data set. Then for any $0<\beta\le 1$, $(M^n,e^{2\beta f}g, e^{\beta f}K)$ is also an AF the initial data set.
\end{cor}
\begin{proof}
Let $\bar g=e^{2\beta f}g$ and $\bar K=e^{\beta f}K$, then at each end, by the assumptions on $f$, it is easy to see that near infinity
$$
|\bar g_{ij}-\delta_{ij}|+r|\p_k\bar g_{ij}|+r|\bar K_{ij}|\le Cr^{-\tau}
$$
for some $C>0$ for all $i,j,k$. On the other hand, it is easy to see that
  $
  2+2\tau>n.
  $
 The results follows from Lemma \ref{l-muJ-1}.

\end{proof}
\begin{cor}\label{c-muJ-1} With the notations as in Lemma \ref{l-muJ-1}, let $\wt g=e^{2f}g$, $\wt K=e^fK$. Let $\wt \mu$, $\wt J$ as in \eqref{e-muJ-1} defined in terms of $\wt g$, then
$$
 \bar\mu
= e^{-2\beta f}\lf((1-\beta)\mu+ \beta e^{2f}\wt\mu+\beta(1-\beta)(n-1)(n-2)  |\nabla_gf|^2\ri),
$$
$$
\ol J^i=e^{-3\beta f}(1-\beta)J^i+\beta e^{3(1-\beta)f}\wt J^i,
$$
and
$$
|\ol J|_{\ol g}\le e^{-2\beta f}\lf((1-\beta)|J|_g+\beta e^{2f}|\wt J|_{\wt g}\ri).
$$
\end{cor}
\begin{proof} By Lemma \ref{l-muJ-1} with $\beta=1$, we have
$$
-2(n-1)\Delta_g f=2e^{2f}\wt \mu-2\mu+(n-1)(n-2)|\nabla f|^2.
$$
By
Lemma \ref{l-muJ-1} again,  we conclude that the first equality in the lemma holds. The second equality can be proved similarly.

To prove the last inequality,
\bee
\begin{split}
|\ol J|^2_{\ol g}=&\ol g_{ij}\ol J^i\ol J^j\\
=&e^{2\beta f}g_{ij}\lf[e^{-3\beta f}(1-\beta)J^i+\beta e^{3(1-\beta)f}\wt J^i\ri]\lf[e^{-3\beta f}(1-\beta)J^j+\beta e^{3(1-\beta)f}\wt J^j\ri]\\
=&e^{-4\beta f}\lf((1-\beta)^2|J|^2_g+\beta^2e^{4f}|\wt J|^2_{\wt g}+2\beta(1-\beta)e^{3f}g_{ij}\wt J^i J^j\ri)\\
\le &e^{-4\beta f}\lf((1-\beta)^2|J|^2_g+\beta^2e^{4f}|\wt J|^2_{\wt g}+2\beta(1-\beta)e^{2f} |\wt J|_{\wt g}|J| \ri)\\
=&e^{-4\beta f}\lf((1-\beta)|J|_g+\beta e^{2f}|\wt J|_{\wt g}\ri)^2.
\end{split}
\eee
The result follows.
\end{proof}

Let us recall the ADM energy momentum vector of an AF initial data set.

\begin{definition}\label{d-definition-EP-1} Let $(M^n,g,K)$ be an AF initial data set. At each end, in AF coordinate chart, the  ADM energy $E$ and the ADM momentum $P$ are defined as follows:
\bee
E=\frac1{2(n-1)\omega_{n-1}}\lim_{r\to\infty}\int_{S_r}\sum_{i,j=1}^n(g_{ij,i}-g_{ii,j})\nu_0^j d\sigma_r
\eee

\bee
P_i=\frac1{(n-1)\omega_{n-1}}\lim_{r\to\infty}\int_{S_r}\sum_{ j=1}^n(K_{ij}-(\tr_gK) g_{ij}) \nu_0^j d\sigma_r
\eee
$i=1,\dots, n$. Here $S_r$ is the coordinate sphere $\{|x|=r\}$, $\nu_0$ is the Euclidean unit outward normal of $S_r$ and $d\sigma_r$ is the area element on $S_r$ induced by the Euclidean metric and $\omega_{n-1}$ is the volume of the unit sphere in $\R^n$.
\end{definition}

\begin{lma}\label{l-EP-1}
Let $(M^n,g, K)$ be an AF initial data set. Fix an end $\mathcal{E}$ of $M$. Let $\tau$, $\e$ be as in Definition \ref{d-initialdata-1}. Let $f$  be  a smooth function such that
 $$
  |f|+r|\nabla_gf|+r|\Delta_gf|^\frac12\le Cr^{-\tau}
  $$
  for some $C>0$ where $r=|x|$ and $x$ is the coordinates of $\mathcal{E}$ in the definition of AF initial data set. Let $0<\beta<1$, and let $\ol g=e^{2\beta f}g$, $\ol K=e^{\beta f}K$; $\wt g=e^{2f}g$, $\wt K=e^{2f}K$. Let $(E,P)$, $(\ol E, \ol P)$, $(\wt E,\wt P)$ be the energy-momentum vectors at $\mathcal{E}$ of  $(M,g,K)$, $(M,\ol g,\ol K)$ and $(M,\wt g,\wt K)$ respectively. Then
\bee
 \ol E=(1-\beta)E+ \beta  \wt E;\ \ \ol P=\wt P=P.
\eee

\end{lma}
\begin{proof} The relation $\ol E=(1-\beta)E+\beta\wt E$ has been observed by Simon \cite{S}, see also \cite{MS} and \cite{G}. On the other hand,
\bee
\begin{split}
\ol K_{ij}-(\tr_{\ol g}\ol K)\ol g_{ij}=&e^{\beta f}\lf(K_{ij}-(\tr_{  g}K)  g_{ij}\ri)\\
=&\lf(K_{ij}-(\tr_{  g}K)  g_{ij}\ri)+(e^{\beta f}-1)\lf(K_{ij}-(\tr_{  g}K)  g_{ij}\ri).
\end{split}
\eee
Since $|e^f-1|\le Cr^{-\tau}$, $|K_{ij}|\le Cr^{-1-\tau}$ and $\tau>\max\{1/2, n-3\}$, which implies $1+2\tau>n-1$, we have $\ol P=P$. Similarly, $\wt P=P$.
\end{proof}
\begin{thm}\label{t-spacetime-pmt-1} Let $(M^n,g,K)$  and $f$ be as in Lemma \ref{l-EP-1}. If $n>3$, we also assume  that $M^n$ is spin.
Let $0<\beta<1$ and let $\ol g=e^{2\beta f}g$, $\ol K=e^{\beta f}K$; $\wt g=e^{2f}g$, $\wt K=e^{2f}K$. Define $\mu, J; \ol\mu, \ol J; \wt \mu, \wt J$ as in Lemma \ref{l-muJ-1} Suppose
$$
(1-\beta)\mu+ \beta e^{2f}\wt \mu\ge (1-\beta)|J|_g+\beta e^{2f}|\wt J|_{\wt g}.
$$
Then for any end $\mathcal{E}$, if $(E,P)$, $(\ol E, \ol P)$, $(\wt E,\wt P)$ be the energy-momentum vectors at $\mathcal{E}$ of  $(M,g,K)$, $(M,\ol g,\ol K)$ and $(M,\wt g,\wt K)$ respectively, then
$$
(1-\beta)E +\beta  \wt E\ge |P|.
$$
If $(1-\beta)E +\beta e^{2f} \wt E=0$ for an end, then $f=0$ and $(M,g)$ can be isometrically embedded in the Minkowski space $(\R^{n+1}, \eta_{ab})$ as an asymptotically flat Cauchy surface with second fundamental form $K$.
\end{thm}
\begin{proof} By Corollary \ref{c-AFinitial-1}, $(M,\ol g,\ol K)$ is an AF initial data set. Suppose
$$
(1-\beta)\mu+ \beta e^{2f} \wt \mu\ge (1-\beta)|J|_g+\beta e^{2f}|\wt J|_{\wt g}.
$$
then $\ol \mu\ge |\ol J|_{\ol g}$ by Corollary \ref{c-muJ-1}. By the spacetime positive mass theorem \cite{SchoenYau,Witten,PT,ChruscielMaerten2006,BeigChruciel}, we see that on each end the energy-momentum vector of $(M,\ol g,\ol K)$ satisfies:
$$
\ol E\ge |\ol P|
$$
which is equivalent to
$$
(1-\beta)E+ \beta  \wt E\ge |P|
$$
by Lemma \ref{l-EP-1}.
Moreover, if $(1-\beta)E +\beta  \wt E=0$ for an end, then $(M,\ol g,\ol K)$ can be isometrically embedded in the Minkowski space  as an asymptotically flat Cauchy surface with second fundamental form $\ol K$. Also $\ol \mu=\ol J=0$. In particular, $M$ has only one end, and $\nabla f=0$ by Corollary \ref{c-muJ-1} and the assumption that
$$
(1-\beta)\mu+ \beta e^{2f}\wt \mu\ge (1-\beta)|J|_g+\beta e^{2f}|\wt J|_{\wt g}.
$$
Hence $f=0$  because $f\to 0$ at infinity. Therefore, $(M,g,K)=(M,\ol g,\ol K)$ and the last assertion of the theorem is true.
\end{proof}
\begin{remark}
Using different version of positive mass theorems, we may obtain corresponding conformal positive mass theorems. For example, if $n<8$ we may use the result \cite{EichmairHuangLeeSchoen} to obtain a corresponding conformal positive mass theorem without assuming the manifold is spin.
\end{remark}
\section{conformal Positive Mass Theorem  for asymptotically hyperbolic manifolds}\label{AH cpmt}

Let us recall the definition  of an asymptotically hyperbolic (AH) manifold and the total mass integral of such an manifold. For simplicity, we assume the manifold has only one end.
We use the definition of an AH manifold as in \cite{ChruscielHerzlich}, see also \cite{Wang, Zhang}.
Let $\mathbb{H}^{n}$ denote the standard hyperbolic space, the metric $ \mathfrak{b}$ is given by
\begin{equation}
 \mathfrak{b}=d\rho^{2}+(\sinh\rho)^{2}h_{0},
\end{equation}
where $h_0$ is the standard metric of the unit sphere $\mathbb{S}^{n-1}$.
Let $\displaystyle{\fe_{i}=\frac{\phi_{i}}{\sinh\rho}}$ and $\displaystyle{\fe_{0}=\frac{\nabla\rho}{|\nabla\rho|}}$ be an orthonormal frame for $ \mathfrak{b}$, where $\{\phi_{i}\}_{i=1,\cdots,n-1}$ is a local orthonormal frame for
$h_{0}$.

\begin{definition}\label{d-AH}  $(M^n,g)$ is called  {\it asymptotically hyperbolic} (AH)
if, outside a compact set, $M$ is diffeomorphic to the exterior of  some geodesic sphere  $\S_{\rho_0}$  in $\H^n$ such that
the metric components  $ g_{ij} =  g(\fe_i,\fe_j)$,  $0 \le i, j \le n-1$,  satisfy
\be
|g_{ij}-\delta_{ij}|=O(e^{-\tau \rho}),\ |\fe_k(g_{ij})|=O(e^{-\tau \rho}), \ |\fe_k(\fe_l(g_{ij}))|=(e^{-\tau \rho})
\ee
for some $\tau>\frac n2$. Moreover, we assume  $\mathcal{S}_g+n(n-1)$ is in
$L^1(\H^n, e^\rho dv_0)$,  where  $dv_0$ is the volume element of $ \mathfrak{b}$ and $\mathcal{S}_g$ is the scalar curvature of $g$.
\end{definition}

Next we want to define the mass integral on an AH manifold. There are equivalent expressions for the mass integral. It is more easy to express the mass integral in terms of the Ricci tensor by the result of Herzlich \cite{Herzlich}. The mass integral is defined as a linear functional on the kernel of the formal adjoint of the linearized scalar curvature $(D\mathcal{S})_{\mathfrak{b}}^*$. To be precise, using the ball model for $\H^n$, so that
$
{B}^n=\{  x \in \R^n|\ |x|<1\}$ and
$\mathfrak{b}_{ij}= \displaystyle{\frac{4}{(1-|x|^2)^2} }\delta_{ij} $. Consider the conformal Killing vector fields
$$ X^{(0)}=x^k\frac{\p}{\p x^k}, \ \   X^{(j)}=\frac{\p}{\p x^j}, \ 1\le j \le n.  $$
 On the other hand, let $ \theta=(\theta^1,\dots,\theta^n)\in \mathbb{S}^{n-1}\subset \R^n$,
    the functions
   \be \label{e-def-Vi}
   V^{(0)}=\cosh \rho , \ \ V^{(j )}=\theta^j  \sinh \rho , \ 1 \le j \le n
   \ee
   which  form a basis of   the kernel of  $(D\mathcal{S})_{\mathfrak{b}}^*$. $X^{(i)}$ and $V^{(i)}$ are related by $\div_{\mathfrak{b}}(X^{(i)}=nV^{(i)}$.
 \begin{definition}\label{d-massint} The mass integral $\mathbf{M}(g)$  for an AH manifold $(M^n,g)$  is the vector $(\mathbf{M}(g) (V^{ (0)}),\dots,\mathbf{M}(g) (V^{ (n)}))$ where

   $$
\mathbf{M}(g) (V^{ (i)})= - c_n\lim_{ \rho  \to\infty}\int_{\S_\rho}G^g (X^{(i)},\nu_g)d\sigma_g .
$$
where $\S_\rho $ is  the  geodesic sphere  of radius $\rho$  centered at $x=0$  in $({B}^n, \mathfrak{b})$,
$\nu_g$ is  the outward unit normal to  $\S_\rho $ with respect to  $g$,   $d\sigma_g$ is the volume element on $\S_\rho $
of the metric induced by $g$, $c_n=\displaystyle{1/(n-1)(n-2)\omega_{n-1}}$ and $\omega_{n-1}$ is the volume of the standard sphere $\mathbb{S}^{n-1}$. Here
\be\label{e-G-1}
G^g=\Ric(g)-\frac12\lf[\mathcal{S}_g+(n-1)(n-2)g\ri].
\ee
\end{definition}

\begin{lma}\label{l-conformal-G-1}
Let $(M^n,g)$ be a Riemannian manifold. Let $f$ be a smooth function on $M$ and let $0<\beta \le 1$. Define $\ol g=e^{2\beta f}g$ and $\wt g=e^{2f}g$. Then
\begin{enumerate}
  \item [(i)]
  $$
\ol{\mathcal{S}}=e^{-2\beta f}\lf\{(1-\beta) \mathcal{S}+\beta e^{2f}\, \wt{\mathcal{S}}+\beta(1-\beta)(n-1)(n-2) |\nabla_g f|^2\ri\}.
$$
where $\mathcal{S}, \ol{\mathcal{S}}, \wt{ \mathcal{S}}$ are the scalar curvatures of $g, \ol g, \wt g$ respectively.

  \item [(ii)]
  \bee
\begin{split}
G^{\ol g}_{ij}=&(1-\beta)G^g_{ij}+\beta G^{\wt g}_{ij}+\frac12(e^{2f}-e^{2\beta f} )(n-1)(n-2)g_{ij}
\\
&+\frac12\beta(\beta-1)(n-2)(n-3)|\nabla_gf|^2g_{ij}+\beta(\beta-1)(n-2)f_if_j.
\end{split}
\eee

\end{enumerate}

\end{lma}
\begin{proof} (i) follows from  Lemma \ref{l-conformal-1} immediately.

(ii) By Lemma \ref{l-conformal-1} and the definition of $G$, we have
\bee
\begin{split}
G^{\ol g}_{ij}=&\ol R_{ij}-\frac12\lf[\ol{\mathcal{S}}+(n-1)(n-2)\ol g\ri]\\
=&R_{ij}-\beta\lf((\Delta_g f)g_{ij}+(n-2)f_{;ij}\ri)+\beta^2(n-2)\lf(f_if_j-|\nabla_gf|^2g_{ij}\ri)\\
&-\frac12\lf[e^{-2\beta f}\lf(\mathcal{S}-2\beta(n-1)\Delta_g f-\beta^2(n-1)(n-2)|\nabla_gf|^2\ri)+(n-1)(n-2)\ri]e^{2\beta f}g_{ij}\\
=&G^g_{ij}+\frac12(1-e^{2\beta f})(n-1)(n-2)g_{ij}+(n-2)\beta \lf((\Delta_gf)g_{ij}-f_{;ij}\ri)\\
&+\frac12\beta^2(n-2)(n-3)|\nabla_gf|^2g_{ij}+\beta^2(n-2)f_if_j
\end{split}
\eee
Similarly,
\bee
\begin{split}
G^{\wt g}_{ij}=& G^g_{ij}+\frac12(1-e^{2  f})(n-1)(n-2)g_{ij}+(n-2) \lf((\Delta_gf)g_{ij}-f_{;ij}\ri)\\
&+\frac12 (n-2)(n-3)|\nabla_gf|^2g_{ij}+ (n-2)f_if_j.
\end{split}
\eee
From the above two relations, the result follows.
\end{proof}
\begin{lma}\label{l-AH-1}Let $(M^n,g)$ be an AH manifold as in Definition \ref{d-AH}. Let $f$ be a smooth function on $M$ such that
$$
|f|+|\fe_kf|+|\fe_l\fe_k f|=O(e^{-\tau \rho})
$$
 for all $0\le l, k\le n-1$. For $0<\beta\le 1$, let $\ol g=e^{2\beta f}g$, $\wt g=e^{2f}g$. Then the mass integrals of $g, \ol g, \wt g$ are related as follows:
 \bee
 \mathbf{M}(\ol g)=(1-\beta)\mathbf{M}(g)+\beta \mathbf{M}(\wt g).
 \eee

\end{lma}
\begin{proof} First note that $\ol g$ and $\wt g$ are also AH. On the other hand, if $\nu_\mathfrak{b}$ is the unit outward  normal of $\S_\rho$ and $d\sigma_{\mathfrak{b}}$ is the volume element of $\S_\rho$ with respect to $\mathfrak{b}$, then (see
   \cite{MiaoTamXie} for example):
\begin{equation}
\lf| \nu_{g}-\nu_{\mathfrak{b}}\ri|_{\mathfrak{b}}= O(e^{-\tau\rho}),
\end{equation}
\begin{equation}
\lf|\frac{d\sigma_{g}}{d\sigma_{\mathfrak{b} }}-1\ri|=O(e^{-\tau\rho}).
\end{equation}
\bee
\lf|\Ric_g(X,Y)-\Ric_{\mathfrak{b}}(X,Y)\ri|= O(e^{-\tau\rho})\lf| X\ri|_{\mathfrak{b}}\lf| Y\ri|_{\mathfrak{ b}},
\eee
\bee
|\mathcal{S}_g-\mathcal{S}_{\mathfrak{b}}|= O(e^{-\tau\rho})
\eee
which imply
$$
|G^g|_{\mathfrak{b}}=O(e^{-\tau\rho}).
$$
because $G^\mathfrak{b}=0$.

Hence if $V^{(i)}$ is as in Definition \ref{d-massint}, we have
$$
\mathbf{M}(g)(V^{(i)})=-c_n\lim_{ \rho  \to\infty}\int_{\S_\rho}G^g (X^{(i)},\nu_\mathfrak{b})d\sigma_\mathfrak{b},
$$
because $|X^{(i)}|_\mathfrak{b}=O(e^\rho)$.
Similarly,
\bee
\mathbf{M}(\ol g)(V^{(i)})=-c_n\lim_{ \rho  \to\infty}\int_{\S_\rho}G^{\ol g} (X^{(i)},\nu_\mathfrak{b})d\sigma_\mathfrak{b},
\eee
and
\bee
\mathbf{M}(\wt g)(V^{(i)})=-c_n\lim_{ \rho  \to\infty}\int_{\S_\rho}G^{\wt g} (X^{(i)},\nu_\mathfrak{b})d\sigma_\mathfrak{b}.
\eee
Since $|\nabla_gf|=O(e^{-\tau \rho})$, by Lemma \ref{l-conformal-G-1}(ii), the result follows.
\end{proof}

\begin{thm}\label{t-cpmt-1}
Let $(M^n,g)$ is an asymptotically hyperbolic spin manifold and let $0<\beta \le 1$, $f$, $\ol g$, $\wt g$ be as in Lemma \ref{l-AH-1}. Suppose
$$
e^{-2\beta f}\lf((1-\beta)\mathcal{S}+\beta e^{2f}\wt{\mathcal{S}}\ri)\ge -n(n-1),
$$
and $\ol {\mathcal{S}}+n(n-1)$, $\wt {\mathcal{S}}+n(n-1)$ are in $L^1(e^\rho dv_0)$. If $M$ has a connected inner boundary $\Sigma$, then we assume $\Sigma$ has   positive Yamabe invariant $Y(\Sigma)$ in the conformal class of $g$ and  the mean curvature $H_g$ with respect to the unit outward normal $\nu$ and $g$   satisfies:

 $$
 e^{-\beta f}\lf(H_g+\beta(n-1)f_{\nu}\ri)\le \lf[\frac{Y(\Sigma)}{(n-1)(n-2)}\lf(\int_\Sigma e^{\beta(n-1)f}d\sigma_g\ri)^{-\frac2{n-1}}+1\ri]^\frac12
 $$
 where $\nu$ is the unit normal pointing insider $M$. Then $(1-\beta)\mathbf{M}(g)+\beta \mathbf{M}(\wt g)$ is future pointing timelike or zero. It is zero then $f=0$, $(M,g)$ has an imaginary Killing spinor field, $\Sigma$ is totally umbilical with respect to $g$ carrying a real Killing spinor. In particular, $(M, g)$ is Einstein.

\end{thm}
\begin{proof} By Lemma \ref{l-conformal-G-1} and the assumption, the scalar curvature of $\ol {\mathcal{S}}$ satisfies:
\be\label{e-AHpmt-1}
\begin{split}
\ol{\mathcal{S}}=&e^{-2\beta f}\lf\{(1-\beta) \mathcal{S}+\beta e^{2f}\, \wt{\mathcal{S}}+\beta(1-\beta)(n-1)(n-2) |\nabla_g f|^2\ri\}\\
\ge& -n(n-1)+\beta(1-\beta)(n-1)(n-2)e^{-2\beta f} |\nabla_g f|^2\\
\ge &-n(n-1).
\end{split}
\ee
Note that if $M$ has an inner boundary, the mean curvature of $\Sigma$ with respect to $\ol g$ is:
$$
H_{\ol g}=e^{-\beta f}\lf(H_g+\beta(n-1)f_{\nu}\ri).
$$
Moreover
$$
|\Sigma|_{\ol g}=\int_\Sigma e^{\beta(n-1)f}d\sigma_g.
$$
We can apply the results of  \cite{Wang,ChruscielHerzlich,H} to the manifold $(M,\ol g)$
to conclude that $\mathbf{M}(\ol g)$ is future timelike or zero. By Lemma \ref{l-AH-1}, we conclude that $(1-\beta)\mathbf{M}(g)+\beta \mathbf{M}(\wt g)$ is future timelike or zero. If it is zero, then by \cite{Wang,ChruscielHerzlich,H},     $\ol g$ has constant scalar curvature $-n(n-1)$. Hence by \eqref{e-AHpmt-1}, $f$ is constant which is zero because $f\to0$ at infinity. Hence $\ol g=g$. The results again follows from \cite{Wang,ChruscielHerzlich,H}.

\end{proof}

\end{document}